\documentclass[11pt,twoside]{article}
\usepackage{acta-info,euler}
\usepackage{graphicx}

\newcommand{\vol}{\textbf{2,} 2 (2010) 210--230}
\newcommand{\CRs}{\CRsubj{%
D.1 
}}   
\newcommand{\amss}{\amssubj{%
68W40, 90C39, 68R10, 05C12
}}     
\newcommand{\keyw}{\keywords{%
dynamic programming, graph theory, shortest path algorithms}}

\begin{document}

\title{Modelling dynamic programming problems by generalized d-graphs}

\maketitle

\oneauthor{\href{http://www.ms.sapientia.ro/~katai_zoltan/}{Zolt\'an K\'atai}}{\href{http://www.emte.ro}{Sapientia Hungarian University of Transylvania}\\\href{http://www.ms.sapientia.ro}{Department of Mathematics and Informatics},\\ Tg. Mure\c s, Romania}{\href{mailto:katai\_zoltan@ms.sapientia.ro}{katai\_zoltan@ms.sapientia.ro} 
}

\short{Z. K\'atai}{Dynamic programming and d-graphs}

\begin{abstract}
In this paper we introduce the concept of generalized d-graph
(admitting cycles) as special dependency-graphs for modelling
dynamic programming (DP) problems. We describe the d-graph
versions of three famous single-source shortest algorithms (The
algorithm based on the topological order of the vertices, Dijkstra
algorithm and Bellman-Ford algorithm), which can be viewed as
general DP strategies in the case of three different class of
optimization problems. The new modelling method also makes
possible to classify DP problems and the corresponding DP
strategies in term of graph theory.
\end{abstract}


\section{Introduction}

Dynamic programming (DP) as optimization method was proposed by
Richard Bellman in 1957 \cite{1}. Since the first book in applied
dynamic programming was published in 1962 \cite{2} DP has become a
current problem solving method in several fields of science:
Applied mathematics \cite{2}, Computer science \cite{3}, Artificial
Intelligence \cite{6}, Bioinformatics \cite{4}, Macroeconomics \cite{13}, etc.
Even in the early book on DP \cite{2} the authors drew attention to
the fact that some dynamic programming strategies can be
formulated as graph search problems. Later this subject was
largely researched. As recent examples: Georgescu and Ionescu
introduced the concept of DP-tree \cite{7}; K\'atai \cite{8} proposed d-graphs
as special hierarchic dependency-graphs for modelling DP problems;
Lew and Mauch \cite{14,15,16}  used specialized Petri Net models to
represent DP problems (Lew called his model Bellman-Net).

All the above mentioned modelling tools are based on cycle free
graphs. As Mauch \cite{16} states, circularity is undesirable if Petri
Nets represent DP problem instances. On the other hand, however,
there are DP problems with ``cyclic functional equation'' (the chain
of recursive dependences of the functional equation is cyclic).
Felzenszwalb and Zabih \cite{5} in their survey entitled \emph{Dynamic
programming and graph algorithms in computer vision} recall that
many dynamic programming algorithms can be viewed as solving a
shortest path problem in a graph (see also \cite{9,11,12}. But,
interestingly, some shortest path algorithms work in cyclic graphs
too. K\'atai, after he has been analyzing the three most common
single-source shortest path algorithms (The algorithm based on the
topological order of the vertices, Dijkstra algorithm and
Bellman-Ford algorithm), concludes that all these algorithms apply
cousin DP strategies \cite{10,17}. Exploiting this observation K\'atai
and Csiki  \cite{12} developed general DP algorithms for discrete
optimization problems that can be modelled by simple digraphs (see
also  \cite{11}). In this paper, modelling finite discrete optimization
problems by generalized d-graphs (admitting cycles), we extend the
previously mentioned method for a more general class of DP
problems. The presented new modelling method also makes possible
to classify DP problems and the corresponding DP strategies in
term of graph theory.

Then again the most common approach taken today for solving
real-world DP problems is to start a specialized software
development project for every problem in particular. There are
several reasons why is benefiting to use the most specific DP
algorithm possible to solve a certain optimization problem. For
instance this approach commonly results in more efficient
algorithms. But a number of researchers in the above mentioned
various fields of applications are not experts in programming.
Dynamic programming problem solving process can be divided into
two steps: (1) the functional equation of the problem is
established (a recursive formula that implements the principle of
the optimality); (2) a computer program is elaborated that
processes the recursive formula in a bottom-up way  \cite{12}. The first
step is reachable for most researchers, but the second one not
necessary. Attaching graph-based models to DP problems results in
the following benefits:

\begin{itemize}\addtolength{\itemsep}{-0.6\baselineskip}
\item it moves DP problems to a well research area: graph theory,
\item it makes possible to class DP strategies in terms of graph theory,
\item as an intermediate representation of the problem (that hides, to some degree,
the variety of DP problems) it enables to automate the programming part
of the problem-solving process by an adequately developed software-tools  \cite{12}, 
\item a general software-tool that automatically solves DP problems
(getting as input the functional equation) should be able to save
considerable software development costs  \cite{16}.
\end{itemize}

\section{Modelling dynamic programming problems}

DP can be used to solve optimization problems (discrete, finite
space) that satisfy the principle of the optimality: \emph{The
optimal solution of the original problem is built on optimal
sub-solutions respect to the corresponding sub-problems.} The
principle of the optimality implicitly suggests that the problem
can be decomposed into (or reduced to) similar sub-problems.
Usually this operation can be performed in several ways. The goal
is to build up the optimal solution of the original problem from
the optimal solutions of its smaller sub-problems. Optimization
problems can often be viewed as special version of more general
problems that ask for all solutions, not only for the optimal one
(A so-called objective function is defined on the set of
sub-problems, which has to be optimized). We will call this
general version of the problem, all-solutions-version.

The set of the sub-problems resulted from the decomposing process
can adequately be modelled by dependency graphs (We have proposed
to model the problem on the basis of the previously established
functional equation that can be considered the output of the
mathematical part and the input of the programming part of the
problem solving process). The vertices (continuous/dashed line
squares in the optimization/all-solutions version of the problem;
see Figures \ref{f2}.a,b,c) represent the sub-problems and directed arcs
the dependencies among them. We introduce the following concepts:

\begin{itemize}\addtolength{\itemsep}{-0.6\baselineskip}
\item \emph{Structural-dependencies:} We have directed arc from vertex A to
vertex B if solutions of sub-problem A \emph{may directly depend}
on solutions of sub-problem B (dashed arcs; see Figure \ref{f2}.a).
\item \emph{Optimal-dependencies:} We have directed arc from vertex A to
 vertex B if the optimal solution of sub-problem A \emph{directly
 depends} on the optimal solution of the \emph{smaller (respect to the optimization
 process)} sub-problem B (continuous arcs; see Figure \ref{f2}.b).
\item \emph{Optimization-dependencies:} We have directed arc from vertex A to
vertex B if the optimal solutions of sub-problem A \emph{may
directly depend} on the optimal solution of the \emph{smaller
(respect to the optimization process)} sub-problem B (dotted arcs;
see Figure \ref{f2}.c).
\end{itemize}

Since structural dependencies reflect the structure of the
problem, the \emph{struc\-tu\-ral-dependencies-graph} can be
considered as input information (It can be built up on the basis
of the functional equation of the problem). This graph may contain
cycles (see Figure \ref{f2}.a). According to the principle of the
optimality the \emph{optimal-depen\-den\-cies-graph} is a rooted
sub-tree (acyclic sub-graph) of the structural-dependencies-graph.
Representing the structure of the optimal solution the
optimal-depen\-den\-cies-graph can be viewed as output information.
Since opti\-mi\-za\-tion-depen\-den\-cies are such structural-depen\-den\-cies
that are \emph{compatible} with the principle of the optimality,
the \emph{optimization-depen\-den\-cies-graph} is a maximal rooted
sub-tree of the structural-dependencies-graph that includes the
optimal-dependencies-tree. Accordingly, the vertices of the
opti\-mi\-za\-tion-dependencies-graph (and implicitly the vertices of
the optimal-dependencies-graph too) can be arranged on levels
(hierarchic structure) in such a way that all its arcs are
directed downward. The original problem (or problem-set) is placed
on the highest level and the trivial ones on the lowest level. We
consider that a sub-problem is structurally trivial if cannot be
decomposed into, or reduced to smaller sub-sub-problems. A
sub-problem is considered to be trivial regarding the optimization
process if its optimal solution trivially results from the input
data. If the structural-dependencies-graph contains cycles, then
completing the hierarchic optimization-dependencies-graph to the
structural-dependencies-graph some added arcs will be directed
upward.

Let us consider, as an example, the following problem: Given the
weighted undirected triangle graph OAB determine

\begin{itemize}\addtolength{\itemsep}{-0.6\baselineskip}
\item all paths from vertex O (origin) to the all vertices (O, A, B) of the graph (Figure \ref{f1}.a),
\item the maximal length paths from vertex O (origin) to the all vertices of the graph
($|OA|=10, |OB|=10, |AB|=100$) (Figure \ref{f1}.b),
\item the minimal length paths from vertex O (origin) to the all vertices of the graph
($|OA|=100, |OB|=10, |AB|=10$) (Figure \ref{f1}.c).
\end{itemize}

\begin{figure}[t!]
\setlength{\unitlength}{1mm}
\begin{picture}(300,40)(0,40)
\put(5,70){\circle*{2}} \put(0,69){A}  \put(5,70){\line(1,0){30}}  
\put(35,70){\circle*{2}} \put(37,69){B} \put(5,70){\line(1,-1){15}}
\put(20,55){\circle*{2}} \put(19,50){O} \put(35,70){\line(-1,-1){15}}
\put(19,43){(a)}

\put(50,70){\circle*{2}} \put(45,69){A}  \put(50,70){\line(1,0){30}}  
\put(80,70){\circle*{2}} \put(82,69){B} \put(50,70){\line(1,-1){15}}
\put(65,55){\circle*{2}} \put(64,50){O} \put(80,70){\line(-1,-1){15}}
\put(62,73){100} \put(53,60){10} \put(73,60){10}
\put(64,43){(b)}

\put(95,70){\circle*{2}} \put(90,69){A}  \put(95,70){\line(1,0){30}}  
\put(125,70){\circle*{2}} \put(127,69){B} \put(95,70){\line(1,-1){15}}
\put(110,55){\circle*{2}} \put(109,50){O} \put(125,70){\line(-1,-1){15}}
\put(109,73){10} \put(95,60){100} \put(119,60){10}
\put(109,43){(c)}
\end{picture}

\caption{The triangle graph}\label{f1}
\end{figure}
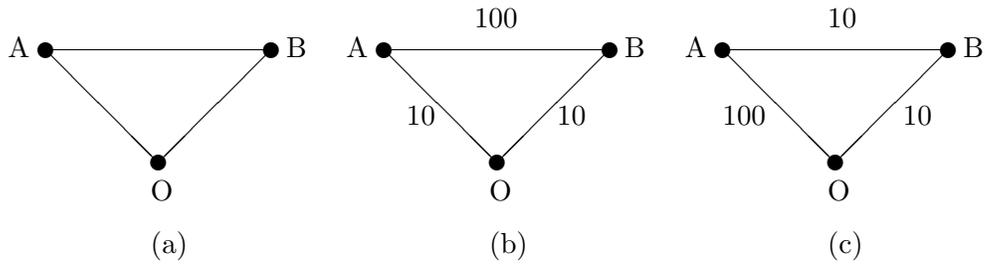

Since path (O,A,B) includes path (O,A) and, conversely, path
(O,B,A) includes path (O,B) the structural-dependencies-graph that
can be attached to the problem is not cycle free (Figure \ref{f2}.a). We
have bidirectional arcs between vertices representing sub-problems
A (determine all paths to vertex A) and B (determine all paths to
vertex B). Since the maximizing version of the problem does not
satisfy the principle of the optimality (the maximal path (O,B,A)
includes path (O,B) that is not a maximal path too), in case b the
optimal-dependencies-tree and the optimization-dependencies-tree
are not defined. Figures \ref{f2}.b and \ref{f2}.c present the optimal- and
optimization-dependencies-graphs attached to the minimizing
version of the problem.

\begin{figure}[htbp]
\centering\includegraphics[width=12cm]{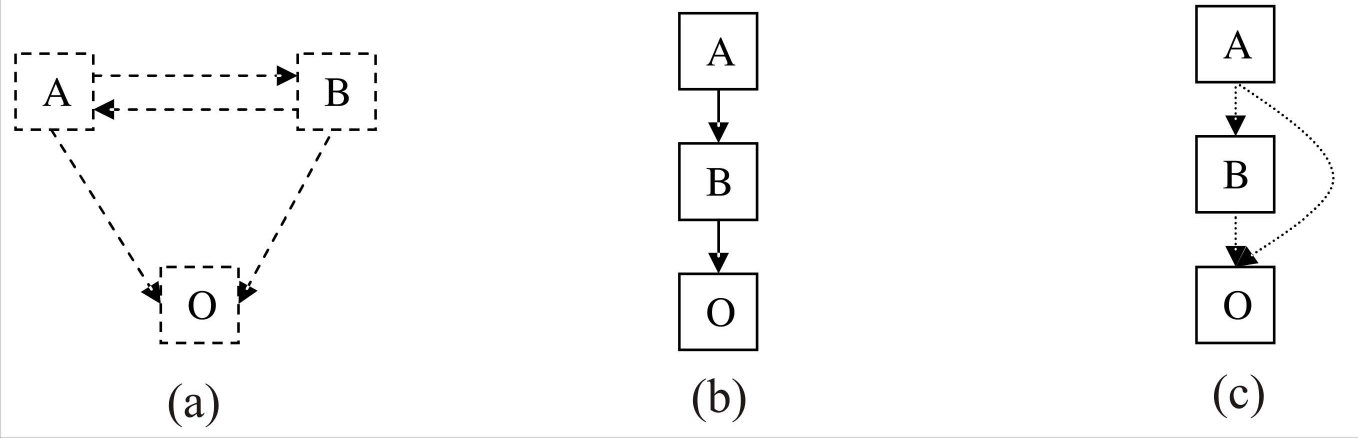}
\caption{Structural/Optimal/Optimization-dependencies-graphs}\label{f2}
\end{figure}

\section{d-graphs as special dependency graphs}

Since decomposing usually means that the current problem is broken
down into \emph{two or more} immediate sub-problems $(1 \rightarrow
N$ dependency) and since this operation can often be performed in
\emph{several} ways, K\'atai \cite{8} introduced d-graphs as special
dependency graphs for modelling such problems. In this paper we
define a generalized form of d-graphs as follows:

\pagebreak

\begin{definition} The connected weighted bipartite finite digraph
$G_{d} (V, E, C)$ is a d-graph if:

\begin{itemize}\addtolength{\itemsep}{-0.6\baselineskip}
\item $V = V_{p} \bigcup V_{d}$ and $E = E_{p} \bigcup E_{d},$ where
\begin{itemize}
\item $V_{p}$ is the set of the p-vertices,
\item $V_{d}$ is the set of the d-vertices,
\item all in/out neighbours of the p-vertices (excepting the source/sink vertices)
 are d-vertices;  each d-vertex has exactly one p-in-neighbour;
 each d-vertex has at least one p-out-neighbour,
\item $E_{p}$ is the set of p-arcs (from p-vertices to d-vertices),
\item $E_{d}$ is the set of d-arcs (from d-vertices to p-vertices),
\end{itemize}
\item function $C: E_{p} \rightarrow \mathbb{R}$ associates a cost to every p-arc. We consider d-arcs of zero cost.
\end{itemize}
\end{definition}

If a d-graph is cycle-free, then its vertices can be arranged on
levels (hierarchic structure) (see Figure \ref{f3}). In \cite{8} K\'atai
defines, respect to hierarchic d-graphs, the following related
concepts: d-sub-graph, d-tree, d-sub-tree, d-spanning-tree,
optimal d-spanning-tree and optimally weighted d-graph.

\vspace{0cm}
\begin{figure}[t!]
\begin{center}
\includegraphics[width=12cm]{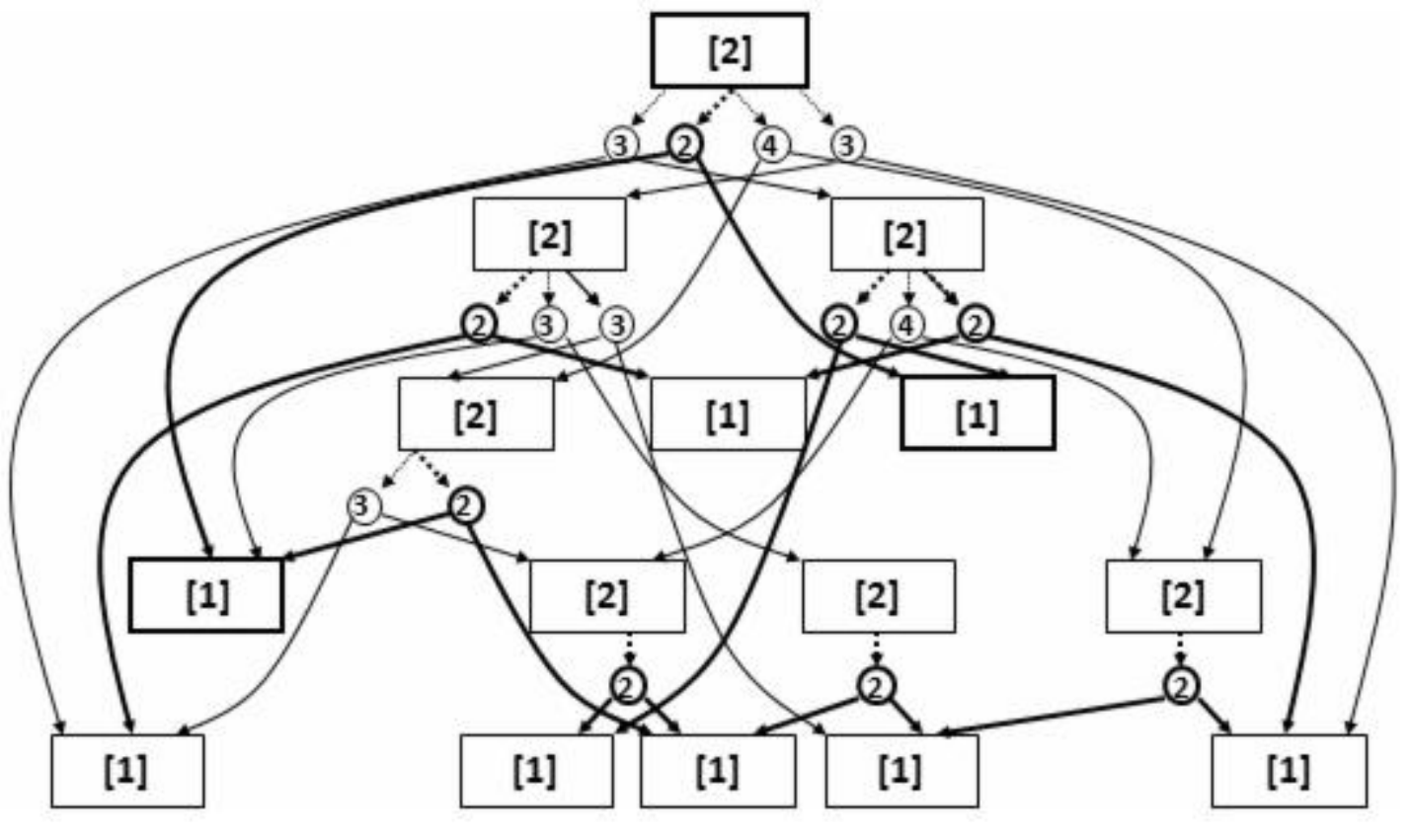}
\end{center}
\caption{Hierarchic d-graph. p- and d-vertices are represented
by rectangles and circles, respectively (We used bolded lines to
emphasize the optimal d-spanning-(sub)trees)}\label{f3}
\end{figure}

\section{Modelling optimization problems by d-graphs}

According to K\'atai \cite{8} a hierarchic d-graph can be viewed as
representing the optimization-de\-pen\-den\-ces-graph corresponding to
the original problem and d-sub-graphs to the sub-problems. Since
there is a one-to-one correspondence between p-vertices and
d-sub-graph \cite{8}, these vertices also represent the sub-problems.
The source p-vertex (or vertices) is attached to the original
problem (or original problem-set), and the sink vertices to the
structurally trivial sub-problems. A p-vertex has as many d-sons
as the number of possibilities to decompose the corresponding
sub-problem into its \emph{smaller immediate} sub-sub-problems. A
d-vertex has as many p-sons as the number of immediate smaller
sub-problems (N) resulted through the corresponding
\emph{breaking-down} step (1 $\rightarrow N$ dependency between the
p-grand\-fa\-ther-problem and the corresponding p-grandson-problems).
We will say that a grandfather-problem is \emph{reduced to} its
grandson-problem if the intermediary d-vertex has a single p-son
(1 $\rightarrow 1$ dependency). Parallel decomposing processes may
result in identical sub-problems, and, consequently, the
corresponding p-vertex has multiple p-grandfathers (through
different d-fathers). Due to this phenomenon the number of the
sub-problems may depend on the size of the input polynomially. The
d-spanning-trees of the d-(sub)graphs represent the corresponding
(sub)solutions, more exactly their tree-structure. The number of
all solutions of the problem usually depends on the size of the
input exponentially.

For example, if a p-vertex has n d-sons, these d-sons have $m_{1},
m_{2}, \ldots, m_{n}$ p-sons, and these p-son-problems have $(r_{1,1},
r_{1,2}, \ldots, r_{1,m1}), (r_{1,1},$ $ r_{1,2}, \ldots, r_{1,m2}),$
$(r_{1,1}, r_{1,2}, \ldots, r_{1,mn})$ solutions, respectively, then
from the $\sum\sum r_{ij}$ solution of the p-grandson-problems
results $\sum\prod r_{ij}$ solution for the common
p-grand\-fa\-ther-problem. The number of solutions exponentially
exceeds the number of sub-problems. The $\sum$-operator reflects
the OR-connection between d-brothers and the $\prod$-operator the
AND-connection between p-brothers.

\section{Dynamic programming strategy on the \\ optimization-dependencies d-graph}

In the case of optimization problems we are interested only in the
\emph{optimal} solution of the original problem. Dynamic
programming means building up the optimal solutions of the
\emph{larger} sub-problems from the optimal solution of the
already solved \emph{smaller} sub-problems (starting with the
optimal solution of the trivial sub-problems). Accordingly, (1) DP
works on the hierarchic optimization-dependencies d-graph that can
be attached to the problem, and (2) it deals with one solution per
sub-problem, with the optimal one (DP strategies usually result in
polynomial algorithms).

In line with this K\'atai \cite{8} defines two weight-functions $(w_{p}:
V_{p} \rightarrow \mathbb{R}, w_{d}: V_{d} \rightarrow \mathbb{R})$
on the sets of p- and d-vertices of the attached hierarchic
d-graph. Whereas the weight of a p-vertex is defined as the
optimum (minimum/maxi\-mum) of the weights of its d-sons, the weight
of a d-vertex is a function (depending on the problem to be
modelled) of the weights of its p-sons. We consider the weight of
a d-vertex to be optimal if is based on optimal the weights of its
p-sons. The optimal weight of a p-vertex (excluding the sink
vertices) is equal with the minimum/maximum of the optimal weights
of its d-sons. The optimal weights of the p-sinks trivially result
from the input data of the problem. Accordingly: the optimal
weights of the p-vertices are computed (1) in \emph{optimal way},
(2) on the basis of the \emph{optimal weights} of their
p-descendents. This means bottom-up strategy. Computing the
optimal weights of the p-vertices we implicitly have their optimal
d-sons (It is characteristic to DP algorithms that during the
bottom-up building process it stores the already computed optimal
p-weights in order to have them at hand in case they are needed to
compute further optimal p-weights. If we also store the optimal
d-sons of the p-vertices, then this information allows a quick
reconstruction of the optimal d-spanning-tree in top-down way
\cite{17,18}).

Defining the costs of the p-arcs as the absolute value of the
weight-difference of its endpoints we get an optimally weighted
d-graph with zero-cost minimal d-spanning-tree. We denote these
kinds of p-arc-cost-functions by $C^{*}$ \cite{8}. Modelling optimization
problems by a d-graphs $G_{d}(V,E,C^{*})$ includes choosing
functions $w_{p}$ and $w_{d}$ in such a way as the optimal weights of
the p-vertices to represent the optimum values of the objective
function respect to the corresponding sub-problems (These
functions can be established on the basis of the functional
equation of the problem; input information regarding the modelling
process).

\begin{proposition}\label{p1} If an optimization problem can be modelled
by a hierarchic d-graph $G_{d}(V,E,C^{*})$ (as we described above),
then it can be solved by dynamic programming.
\end{proposition}

\begin{proof} Since in an optimally weighted d-graph d-sub-trees
of an optimal d-spanning-tree are also optimal d-spanning-trees
respect to the d-sub-graphs defined by their root-vertices,
computing the optimal p- and d-weights according to a reverse
topological order of the vertices (based on
optimization-dependencies) implicitly identifies the optimal
d-spanning-tree of the d-graph. This bottom-up strategy means DP
and the optimal solution of the original problem will be
represented by the weight of the source vertex (as value) and by
the minimal d-spanning-tree (as structure).
\end{proof}

Computing the optimal weight of a p-vertex (expecting vertices
representing trivial sub-problems) can be implemented as a gradual
updating process based on the weights of its d-sons. The weights
of p-vertices representing trivial sub-problems receive as
starting-value their input optimal value. For other p-vertices we
choose a proper starting-value according to the nature of the
optimization problem (The weights of d-vertices are recomputed
before every use). We define the following types of updating
operations along p-arcs (if the weight of a certain d-son is
``better'' than the weight of his p-father, then the father's weight
is replaced with the son's weight):

\begin{itemize}\addtolength{\itemsep}{-0.6\baselineskip}
\item \emph{Complete:} based on the optimal value of the corresponding d-son.
\item \emph{Partial:} based on an intermediate value of the corresponding d-son.
\item \emph{Effective:} effectively improves the weight of the corresponding p-vertex.
\item \emph{Null:} dose not adjusts the weight of the corresponding p-vertex.
\item \emph{Optimal:} sets the optimal weight for the corresponding p-vertex.
Optimal updates are complete and effective too.
\end{itemize}


\section{d-graph versions of three famous single-source \\ shortest-path algorithms}

As we mentioned above, K\'atai concludes that the three famous
single-source shortest-path algorithms in digraphs (The algorithm
based on the topological order of the vertices, Dijkstra algorithm
and Bellman-Ford algorithm) apply cousin DP strategies \cite{10,17}.
The common representative core of these DP algorithms is that the
optimal weights (representing the optimal lengths to the
corresponding vertices) are computed on account of updating these
values along the arcs of the shortest-paths-tree according to
their topological order (optimal-updating-sequence). Since this
optimal tree is unknown (it represents the solution of the
problem) all the three algorithms generate updating-sequences
which contain, as subsequence, an optimal-updating-sequence
necessary for the dynamic programming building process. The basic
difference among the three algorithms is the way they generate a
proper arc-sequence and the corresponding updating-sequence.

\vspace{0cm}

\begin{figure}[t!]
\centering\includegraphics[height=16cm]{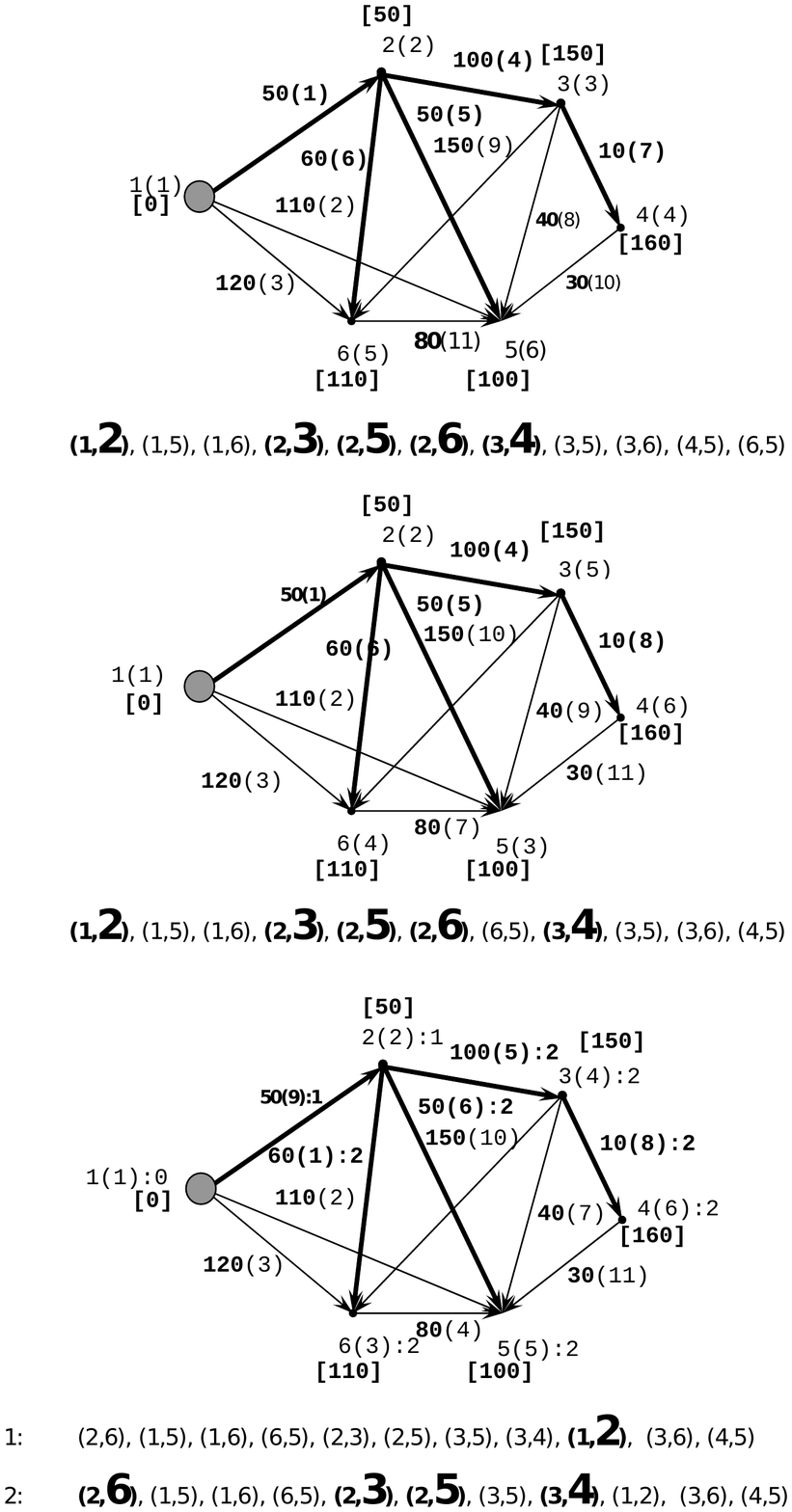}
\caption{The strategies of the (a) Topological, (b) Dijkstra and
(c) Bellman-Ford algorithms (we bolded the optimal-arc-subsequence
of the generated arc-sequences)} \label{f4}
\end{figure}

In case the input digraph is acyclic we get a proper arc-sequence
by ordering all the arcs of the graph topologically (this order
can even be determined in advance).  Dijkstra algorithm (working
in cyclic graphs too, but without negative weighted arcs)
determines the needed arc-sequence on the fly (parallel with the
updating process). After the current weight of the next closest
vertex has been confirmed as optimal value (greedy decision), the
algorithm performs updating operations along the out-arcs of this
vertex (This greedy choice can be justified as follows: if other
out-neighbours of the growing shortest-paths-tree are farther -
from the source vertex - than the currently closest one, then
through these vertices cannot lead shortest paths to this).
Bellman-Ford algorithm (working in cyclic graphs with negative
weighted arcs too, but without feasible negative weighted cycles)
goes through (in arbitrary order) all the arcs of the graph again
and again until the arc-sequence generated in this way finally
will contains, as sub-sequence, an optimal-updating-sequence (see
Figure \ref{f4}, \cite{10}). The following d-graph algorithms implement DP
strategies that exploit the core idea behind the above described
single-source shortest-paths algorithms.

\subsection{Building-up the optimization-dependencies d-graph in \\ bottom-up way}

Our basic goal is to perform updating operation along the p-arcs
of the optimal-dependencies-tree according to their reverse
topological order. We will call such arc sequences
optimal-arc-sequence and the corresponding updating sequences
optimal-updating-sequence. An optimal-updating-sequence surely
results in building up the optimal value representing the optimal
solution of the problem.  Since the optimal-dependencies-tree is
unknown (it represents the structure of the optimal solution to be
determined), we should try to elaborate \emph{complete arc
sequences} that includes the desired optimal-updating-sequence
(gratuitous updating operations have, at the worst, null effects).

We introduce the following colouring-convention:

\begin{itemize}\addtolength{\itemsep}{-0.6\baselineskip}
\item Initially all vertices are white.
\item A p-vertex changes its colour to grey after the first
attempt to update its weight. d-vertices automatically change
their colour to grey if they do not have any more white p-sons.
\item When the weight of a vertex riches its optimal value
its colour is automatically changed to black.
\end{itemize}

We are facing a gratuitous updating operation if:
\begin{itemize}\addtolength{\itemsep}{-0.6\baselineskip}
\item along the corresponding p-arc was previously performed a complete update,
\item the corresponding p-father is already black,
\item the corresponding d-son is still grey or white.
\end{itemize}

Since the optimal values of trivial sub-problems automatically
results from the input data of the problem, the corresponding
p-vertices are automatically coloured with black.

The following propositions can be viewed as theoretical support
for the below strategies that build up the optimal-dependencies
d-graph (on the basis of the structural-dependencies-graph that
can be considered input information) level-by-level in bottom-up
way (At the beginning all p-vertices are places at level 0. All
effective updates along the corresponding p-arcs move their p-end
to higher level than the highest p-son of their d-end.).

\begin{proposition}\label{p2} If the structural-dependencies d-graph
attached to an optimization problem that satisfies the principle
of the optimality has no black p-sources, then there exists at
least one p-arc corresponding to an \emph{effective complete}
updating operation.
\end{proposition}

\begin{proof} Since the optimization problem satisfies the
principle of the optimality the optimal-updating-sequence there
exists and continuously warrants (while no black p-sources still
exist) the existence of optimal updating operations, which are
effective and complete too.
\end{proof}

\begin{proposition}\label{p3} Any p-arcs sequence (of the
structural-dependencies d-graph attached to an optimization
problem that satisfies the principle of the optimality) that
continuously applies non-repetitive complete updates (while such
updating operations still exist) warrants that all p-sources
become black-coloured. These p-arcs sequences contain arcs
representing optimization-depen\-den\-cies and surely include an
optimal-arc-sequence.
\end{proposition}

\begin{proof} Since the optimization problem satisfies the
principle of the optimality the optimal-updating-sequence there
exists and warrants the continuous existence of optimal updating
operations (which are also effective and complete) while no black
p-sources still exist. Accordingly any p-arcs sequence that
con\-tinuously applies non-repetitive complete updates includes an
optimal-arc-sequence, and consequently results in colouring all
p-sources with black.
\end{proof}

\begin{proposition}\label{p4} If the structural-dependencies d-graph
attached to an optimization problem that satisfies the principle
of the optimality is cycle-free, then any reverse topological
order of the all p-arcs continuously applies non-repetitive
complete updates, and consequently, results in building up the
optimal solution of the problem.
\end{proposition}

\begin{proof} Since the colours of the d-vertices surely become
black after all their p-sons have already become black, any
reverse topological \mbox{order of all p-arcs con}\-tinuously applies
non-repetitive complete updates. \mbox{According to the previous}
proposition these arc-sequences surely include an
optimal-arc-sequence, and consequently results in building up the
optimal solution of the problem.  
\end{proof}

\begin{proposition}\label{p5} If an optimization problem satisfies the
principle of the optimality, then there exists a finite multiple
complete arc-sequence of the attached structural-dependencies
d-graph that includes an optimal-arc-sequence, and consequently,
the corresponding updating-sequence results in building up the
optimal solution of the problem.
\end{proposition}

\begin{proof} The existence of such an arc-sequence immediately
results from the facts that: (1) Any complete arc-sequence
contains all arcs of the optimal-de\-pen\-den\-cies-tree; (2) The
optimal-dependencies-tree is finite. If we repeat a complete
arc-sequence that includes the arcs of the
optimal-dependencies-tree according to their topological order
(worst case), \mbox{then we need as many upda}\-ting-tours as the number of
the p-arcs of the \mbox{optimal-dependencies-tree is}.
\end{proof}

\subsubsection{Algorithm d-TOPOLOGICAL}

If the structural-dependencies d-graph attached to the problem is
cycle free (called: structurally acyclic DP problems), then this
input graph can also be viewed as optimization-dependencies-graph.
Considering a reverse topological order of the \emph{all}
vertices, all updating operations (along the corresponding
p-arc-sequence) will be complete (see Proposition \ref{p4}).
Additionally, along the arcs of the optimal d-spanning-tree we
have optimal updates. Accordingly, this algorithm (called
d-TOPOLOGICAL) results in determining the optimal solution of the
problem.

\subsubsection{Algorithm d-DIJKSTRA}

If the structural-dependencies d-graph contains cycles a proper
vertices order involving complete updates along the corresponding
p-arc-sequence cannot be structurally established. In this case we
should try to build up the optimization-dependencies d-graph (more
exactly a reverse topological order of its all p-arcs) on the fly,
parallel with the bottom-up optimization process.

Implementing a sequence of continuous complete updates presumes to
identify at each stage of the building process the black
d-vertices based on which we have not performed complete updating
operations (Proposition \ref{p2} guaranties that such d-vertices exist
continuously). A d-vertex is black only if all its p-sons are
already black. Consequently, the basic question is: Can we
identify the black p-vertices at each stage of the building
process? As we mentioned above a p-vertex is certainly black after
we have performed complete updates based on \emph{all} its d-sons
(The last effective update was performed on the basis of optimal
d-son). Algorithms based on the topological order of the all arcs
exploit this \emph{structural} condition. However, a p-vertex may
have become black before we have performed complete updating
operation along all its p-out-arcs. Conditions making perceptible
such black p-vertices may also be deduced from the principle of
the optimality.  For example, if the DP problem has a greedy
character too, then it may work the following condition: the
``best'' d-vertex (having \emph{relatively} optimal weight) among
those based on which we have not performed complete updating
operations can be considered black (Called: Cyclic DP problems
characterized by greedy choices). Since Dijkstra algorithm
applies this strategy, we call this algorithm: d-DIJKSTRA.

\subsubsection{Algorithm d-BELLMAN-FORD}

If we cannot establish one complete arc-sequence including an
optimal-arc-sequence (we will call such problems: DP problems
without 'negative cycles'), we are forced to repeat the
updating-tour along a complete (even arbitrary) arc-sequence of
the input graph (structural-dependencies d-graph) until this
\emph{multiple arc-sequence} will include the desired optimal
updating sequence (see Proposition \ref{p5}). An extra tour without any
effective updates indicates that the optimal solution has been
built up. If the arbitrary arc-sequence we have chosen includes
the arcs of the optimal-dependencies-tree in topological order
(worst case), then we need as many updating-tours as the number of
the p-arcs of the optimal-dependencies-tree is. Since Bellman-Ford
algorithm applies this strategy, we call this algorithm:
d-BELLMAN-FORD.


\begin{center}
\begin{figure}[th!]
\centering\includegraphics[scale=0.7]{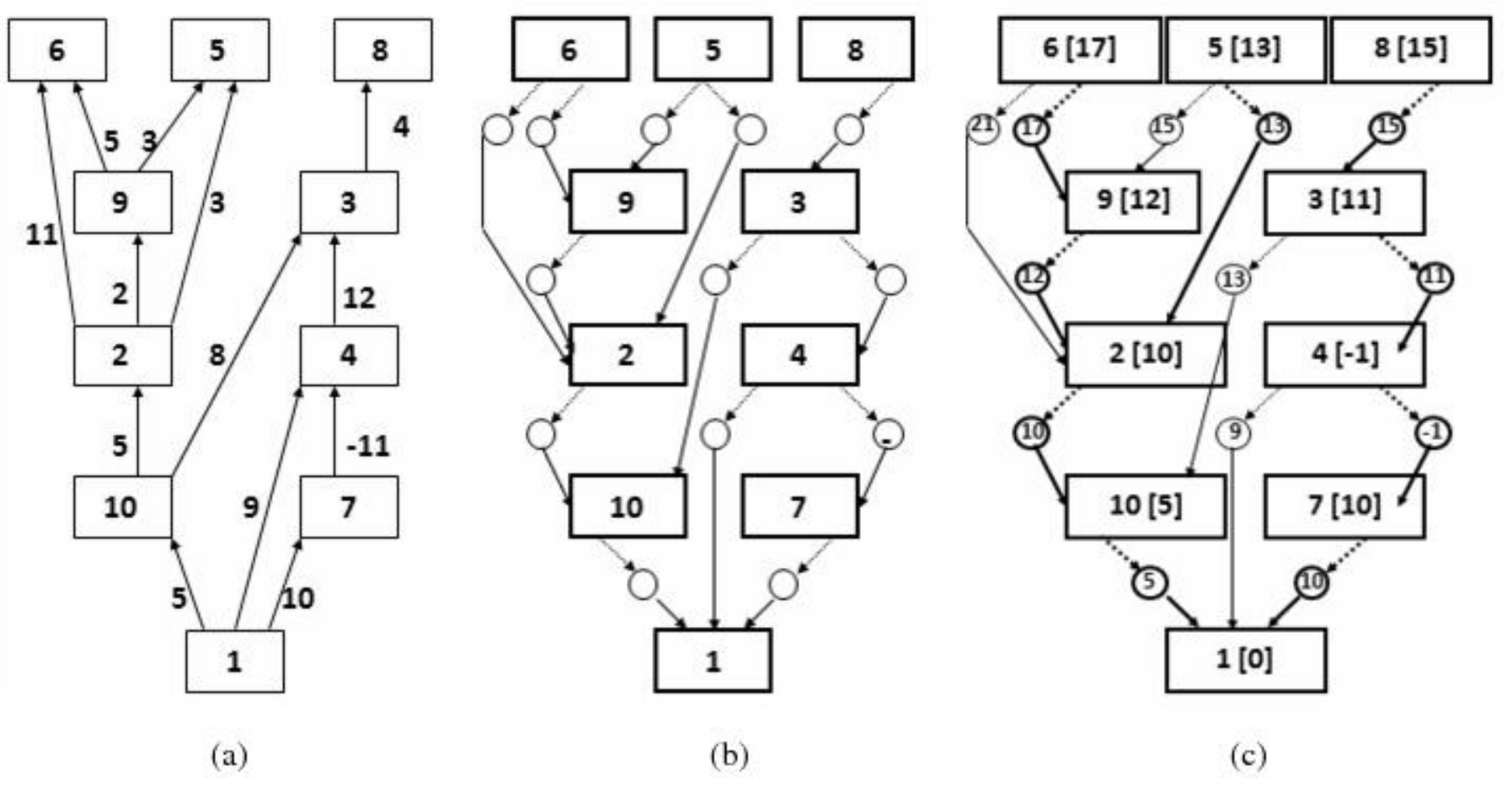}  
\caption{(a) Acyclic digraph; (b) Structural-dependencies
d-graph; (c) Optimally weighted optimization-dependencies d-graph
(bolded lines represent the arcs of the optimal-dependencies
d-graph)}\label{f5}
\end{figure}

\begin{figure}[th!]
\centering\includegraphics[scale=0.6]{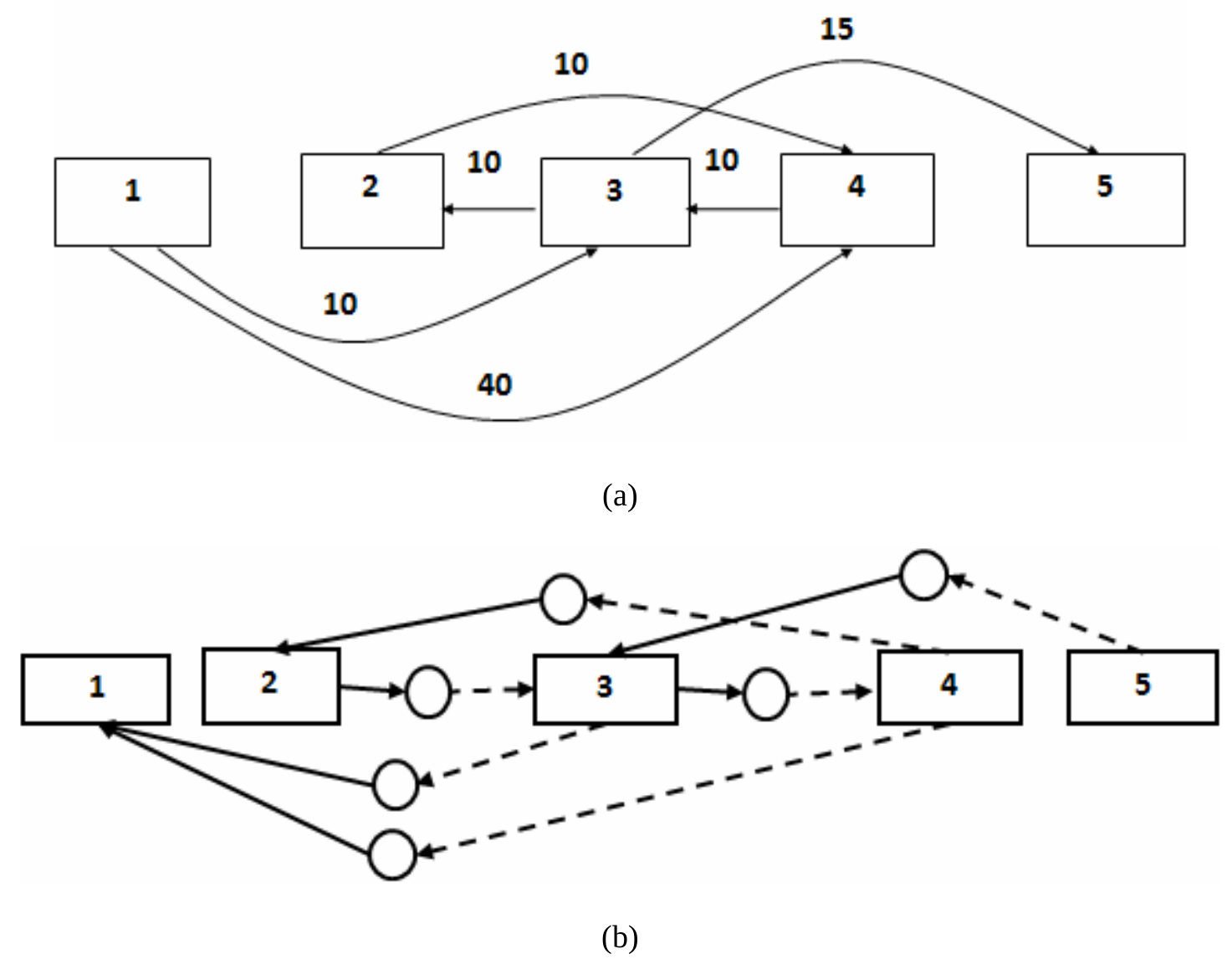} 
\end{figure}
\end{center}

\vspace*{3cm}
 
\begin{figure}[t!]
\centering\includegraphics[height=14cm,width=12cm]{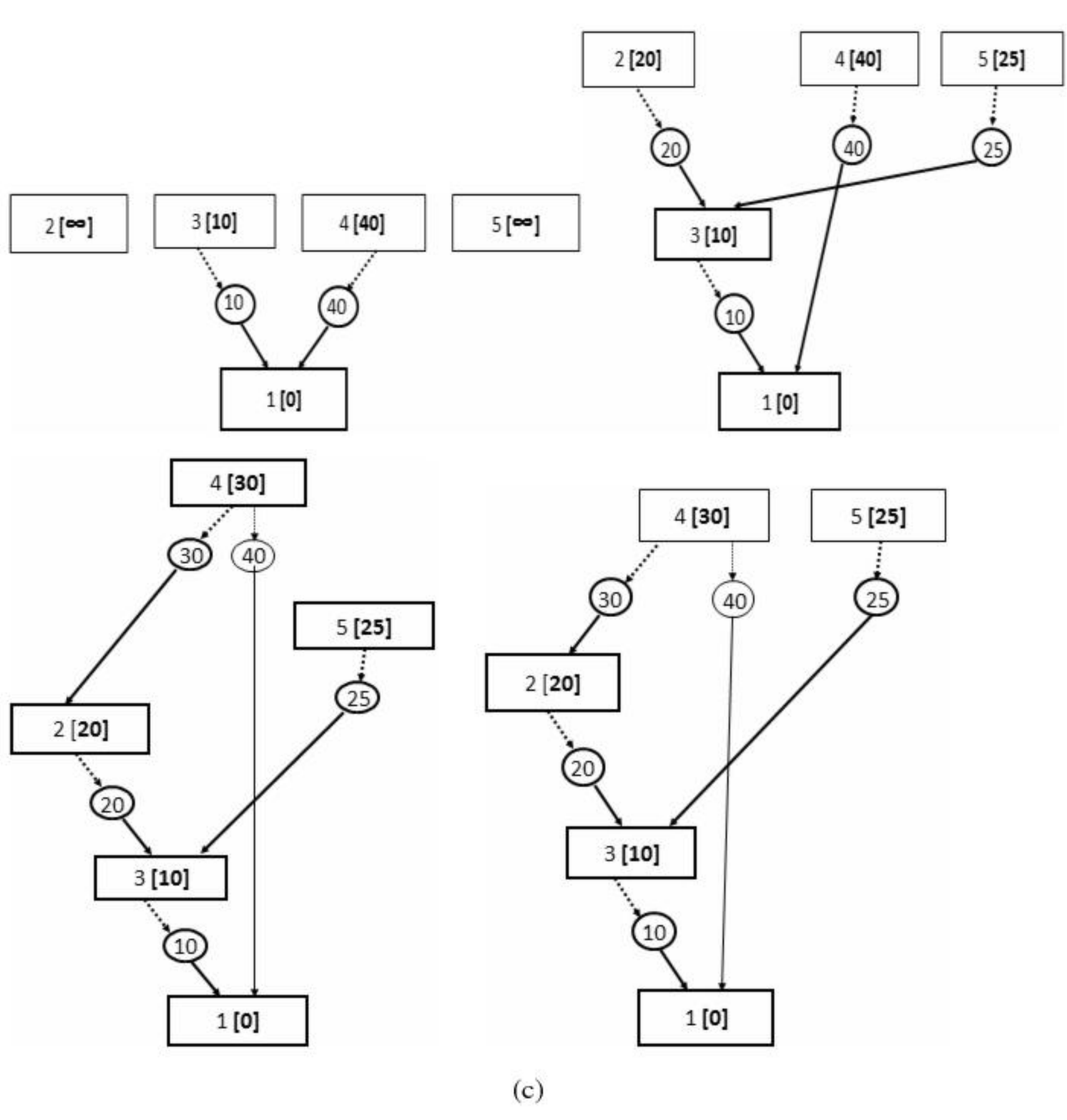} \\
\caption{(a) Cyclic digraph without negative weighted arcs; (b)
Cyclic structural-dependencies d-graph; (c) The bottom-up
building-up process of the optimally weighted
optimization-dependencies d-graph (bolded lines represent the arcs
of the optimal-dependencies d-graph)} \label{f6}
\end{figure}

\vspace*{3cm}
\begin{figure}[th!]
\centering\includegraphics[scale=0.7]{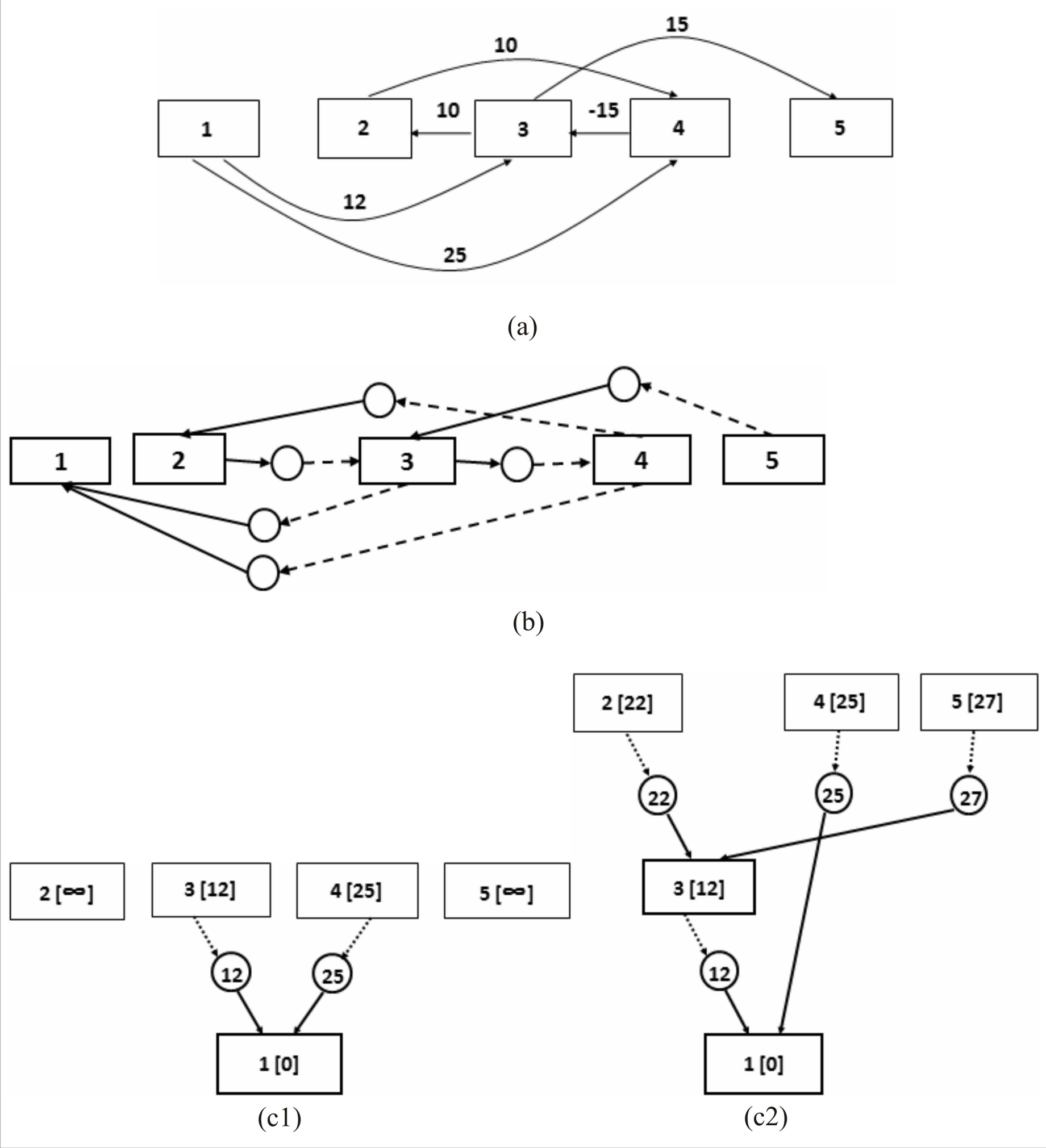}  
\end{figure}

\begin{figure}[ht!]
\centering\includegraphics[height=14cm]{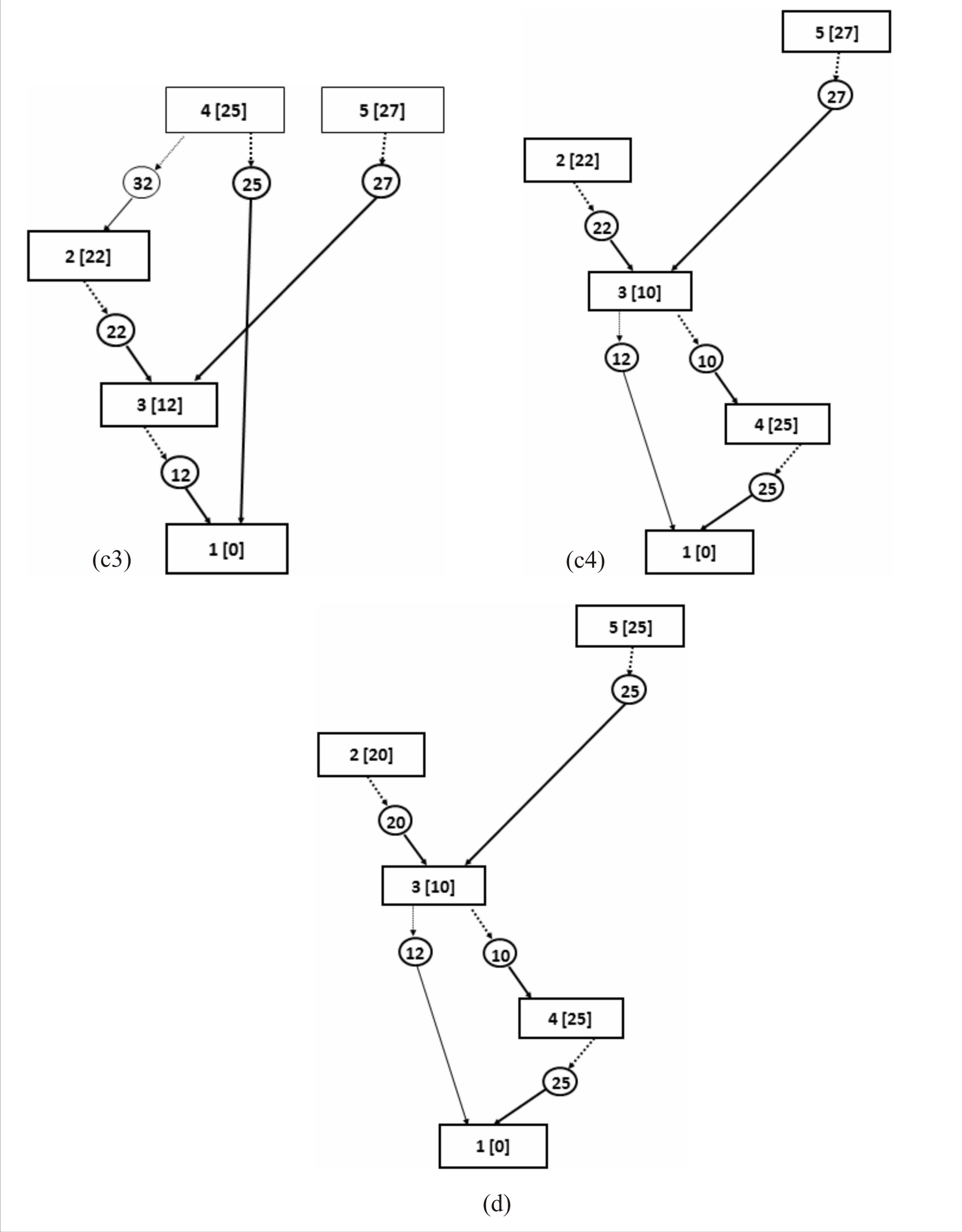}
\caption{(a) Cyclic digraph with negative weighted arcs, but
without negative cycles; (b) Cyclic structural-dependencies
d-graph; The bottom-up building-up process of the optimally
weighted optimization-dependencies d-graph (bolded lines represent
the arcs of the optimal-dependencies d-graph): (c1--c4) first
updating-tour, (d) second updating-tour}\label{f7}
\end{figure}

\vspace*{6cm}
\subsubsection{A relevant sample problem}

As an example we consider the single-source shortest problem:
Given a weighted digraph determine the shortest paths from a
source vertex to all the other vertices (destination vertices).
The attached figures (see Figures \ref{f5}, \ref{f6}, \ref{f7}) illustrate the level by
level building process of the optimization-dependencies d-graph
concerning to the algorithms d-TOPOLOGICAL, d-DIJKSTRA and
d-BELLMAN-FORD (Regarding this problem we have only $1 \rightarrow
1$ dependencies between neighbour p-vertices).

\section{Conclusions}

Introducing the generalized version of d-graphs we received a more
effective tool for modelling a larger class of DP problems
(Hierarchic d-graphs introduced in \cite{8} and Petri-net based models
\cite{14,15,16} work only in the case of structurally acyclic problems;
Classic digraphs \cite{11,12}  can be applied when during the
decomposing process at each step the current problem is reduced to
only one sub-problem). The new modelling method also makes
possible to classify DP problems (Structurally acyclic DP
problems; Cyclic DP problems characterized by greedy choices; DP
problems without 'negative cycles') and the corresponding DP
strategies (d-TOPOLOGICAL, d-DIJKSTRA, d-BELLMAN-FORD) in term of
graph theory.

If we have proposed to develop a general software-tool that
automatically solves DP problems (getting as input the functional
equation) we should combine the above algorithms as follows:

\begin{itemize}\addtolength{\itemsep}{-0.6\baselineskip}
\item We represent explicitly the d-graph described implicitly by the
functional equation.
\item We try to establish the reverse topological order of the vertices by a DFS like algorithm (d-DFS).
This algorithm can also detect possible cycles.
\item If the graph is cycle free, we apply algorithm d-TOPOLOGICAL, else we try to
apply algorithm d-DIJKSTRA.
\item If no mathematical guarantees that
we reached the optimal solution, then choosing as complete
arc-sequence for algorithm d-BELLMAN-FORD the arc-sequence
generated by algorithm d-DIJKSTRA (completed with unused arcs) in
the first updating-tour we verify the d-DIJKSTRA result. We repeat
the updating tours until no more effective updates.
\end{itemize}

Such a software-application should be able to save considerable
software development costs.

\section{Acknowledgements}

This research was supported by the Research Programs Institute of Sapientia Foundation, Cluj, Romania.

\bigskip
\rightline{\emph{Received:  September 1, 2010  {\tiny \raisebox{2pt}{$\bullet$\!}}  Revised: November 10, 2010}}

\end{document}